\documentclass[conference, pdftex]{IEEEtran}

\usepackage{cite}
\usepackage{amsmath,amssymb,amsfonts}
\usepackage{algorithmic}
\usepackage{textcomp}
\usepackage{balance, hyperref, cleveref}
\usepackage{listings, xcolor, amsthm}
\usepackage{array, multirow, booktabs, tikz}
\usepackage{diagbox}
\usepackage{flushend}
\usepackage[ruled, linesnumbered]{algorithm2e}
\usepackage{orcidlink}
\usepackage{graphicx}

\newtheorem{theorem}{Theorem}
\usepackage{xcolor}

\def\BibTeX{{\rm B\kern-.05em{\sc i\kern-.025em b}\kern-.08em
    T\kern-.1667em\lower.7ex\hbox{E}\kern-.125emX}}

\newcommand{\tool}{CountChain}

\begin{document}

\title{\tool: A Decentralized Oracle Network for Counting Systems}

\author{
\IEEEauthorblockN{1\textsuperscript{st} Behkish Nassirzadeh}
\IEEEauthorblockA{\textit{University of Waterloo} \\
Waterloo, Canada \\
bnassirz@uwaterloo.ca \orcidlink{0000-0002-3227-1409}}
\and

\IEEEauthorblockN{2\textsuperscript{nd} Stefanos Leonardos}
\IEEEauthorblockA{\textit{King's College London} \\
London, UK \\
stefanos.leonardos@kcl.ac.uk \orcidlink{0000-0002-1498-1490}}
\and

\IEEEauthorblockN{3\textsuperscript{rd} Albert Heinle}
\IEEEauthorblockA{\textit{CoGaurd} \\
Waterloo, Canada \\
albert.heinle@googlemail.com}
\and

\IEEEauthorblockN{4\textsuperscript{th} Anwar Hasan}
\IEEEauthorblockA{\textit{University of Waterloo} \\
Waterloo, Canada \\
ahasan@uwaterloo.ca \orcidlink{0000-0003-4103-7945}}
\and

\IEEEauthorblockN{5\textsuperscript{th} Vijay Ganesh}
\IEEEauthorblockA{\textit{Georgia Institute of Technology} \\
Atlanta, USA \\
vganesh@gatech.edu \orcidlink{0000-0002-6029-2047}}
}

\maketitle

\begin{abstract}
Blockchain integration in industries like online advertising is hindered by its connectivity limitations to off-chain data. These industries heavily rely on precise counting systems for collecting and analyzing off-chain data. This requires mechanisms, often called oracles, to feed off-chain data into smart contracts. However, current oracle solutions are ill-suited for counting systems since the oracles do not know when to expect the data, posing a significant challenge.

To address this, we present \tool, a decentralized oracle network for counting systems. In \tool, data is received by all oracle nodes, and any node can submit a proposition request. Each proposition contains enough data to evaluate the occurrence of an event. Only randomly selected nodes participate in a game to evaluate the truthfulness of each proposition by providing proof and some stake. Finally, the propositions with the outcome of True increment the counter in a smart contract. Thus, instead of a contract calling oracles for data, in \tool, the oracles call a smart contract when the data is available. Furthermore, we present a formal analysis and experimental evaluation of the system's parameters on over half a million data points to obtain optimal system parameters. In such conditions, our game-theoretical analysis demonstrates that a Nash equilibrium exists wherein all rational parties participate with honesty.
\end{abstract}

\begin{IEEEkeywords}
Counting System, Blockchain, Decentralized Oracle, Smart Contract 
\end{IEEEkeywords}

\section{Introduction}
\label{intro}

Counting systems are mechanisms that keep track of the occurrence of an event. These systems are essential for industries like online advertising, where numerous transactions transpire quickly. In most counting systems, precision is vital. However, conventional counting system designs are prone to inaccuracy, single point of failure, and trust issues. In online advertising, independent data recording by one or multiple trusted parties may result in discrepancies, costing honest parties billions of dollars\cite{AdFraud}. These challenges can be mitigated using a blockchain-based approach, which leverages smart contracts and Decentralized Oracle Networks (DONs)\cite{chainlinkblog}\cite{coindesk}.

Blockchain is subject to a significant limitation, the \emph{blockchain oracle problem}, which is the smart contracts' limited ability to interact with off-chain data. This limitation is one of the main barriers to overcome if Blockchains are to achieve mass adoption across various markets\cite{coindesk}\cite{chainl}\cite{augur}. To fully address the oracle problem, a decentralized oracle solution is required to eliminate data manipulation, inaccuracy, and single point of failure. A DON accomplishes this by bringing together numerous independent oracle nodes and multiple reliable data sources to establish decentralization from end to end\cite{chainl}.

\begin{figure}[t]
\centerline{\includegraphics[width= 0.6 \textwidth/2]{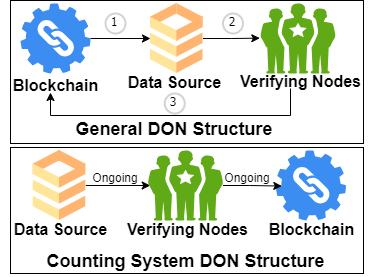}}
\centering
\caption{General DON vs. Counting System DON Design}
\label{problem}
\end{figure}

\label{problemstatement}
\textbf{Problem Statement:} Most of the current DONs fail to provide a reliable solution for high-scale and precise counting systems. This is a significant challenge to overcome if billion-dollar industries like online advertisement that rely heavily on precise counting systems are to adopt blockchain. In the current DONs, smart contracts could call an outside source to gather off-chain data. The data is fed to multiple oracles in a DON to be verified. Once the trustworthiness of the data is validated, the data is fed back to the smart contract. Therefore, a smart contract initiates the process and knows when to expect the off-chain data and some error is generally tolerated.

In counting systems, events can be continuous and occur unpredictably. Therefore, the smart contract does not necessarily know when to expect the data. The data source collects the information and sends it to the verifying oracle nodes in a DON. Once the DON validates the trustworthiness of the source and data, it triggers a condition in the smart contract, which then executes and receives the data. Figure ~\ref{problem} illustrates the difference between the existing DONs and the one needed for a counting system.

In context, a DON like ChainLink\cite{chainl} can be utilized to acquire the current ETH price. This involves creating a Service Level Agreement (SLA) contract to specify request terms. Oracles are then tasked with providing the data to an aggregating contract, which determines and reports the outcome to the requester. Finally, oracles are rewarded financially for their participation. In this case, discrepancies are resolved using techniques such as taking the median of results.

This is not the case for most counting systems. For instance, tracking ad impressions on a website involves multiple invisible pixels sending data to predefined endpoints, and data loss occurs if they become unavailable. For this, one can use ChainLink to implement an SLA contract to request the number of ad views in a certain period, and the oracles will respond with a number. Then, the Aggregating contract reports the aggregation of the results. The main problem with this approach is that it is challenging to resolve discrepancies between the oracles. For instance, if three oracles report ten views and two report nine, the network cannot decide if the two oracles missed counting the same or different ad views. The second approach is implementing SLA as a True/False proposition (i.e., whether user X viewed the ad on website Y). This approach is prone to a single point of failure (the requester might miss some data). The last approach combines the two where the SLA contract requests for the number of ad views in a certain period, one oracle submits a proposition, and the other oracles agree or disagree. This approach requires a new consensus protocol in addition to the current ChainLink protocol so that the oracles do not cheat or lazy vote.

Additionally, scalability can be a challenge, and providing a financial reward for every request requires many transactions, which increases the system's costs. Also, if the event is ongoing, the SLA needs to be implemented as such, which can increase the design's complicity. Finally, in most counting systems, precision is crucial as compensation is based on the exact count. For online advertising, the impression discrepancy caused by traditional techniques is estimated to cost honest players up to 20\% of their revenue annually (up to \$80B), underscores the urgent need for more reliable solutions\cite{AdFraud}\cite{20p}. Current approaches, such as relying on median values or powerful parties like Google, often fall short. Moreover, in most cases, additional information like user IDs and timestamps is required for accurate record reconciliation. This is why the current DONs are inadequate.

\textbf{Approach and Results:} We propose \tool, a DON designed for counting systems. All the blockchain nodes in the system receive the counting data from a source. The nodes can be a \emph{submitter} or a \emph{verifier}. The submitter pays a fee to submit boolean \emph{propositions} to the system. These propositions contain enough information for verifiers to validate their correctness. Verifiers are chosen randomly by the network, and the verifiers must vote True or False and provide proof and a fee as a stake. The majority of the verifiers' votes decide the proposition's outcome. The nodes receive points for raising a proposition or a valid vote for a True proposition. The accumulated points are then converted to financial rewards. Meanwhile, a dishonest vote/proposition results in the node losing money and points. The system is designed to encourage players to vote honestly. This system can be implemented separately or on top of the existing platforms like ChainLink. Hence, in this paper, we make three major contributions:
\begin{enumerate}
    \item The design and implementation of CountChain, a DON tailored for counting systems that offers trust, scalability, and resistance to common attacks.
    \item A game-theoretical analysis that demonstrates the existence of a Nash equilibrium where honesty is the optimal strategy.
    \item Extensive experiments, involving over half a million data instances, to demonstrate \tool's feasibility and resilience to Sybil attacks.
\end{enumerate}

\begin{figure*}[t]
\includegraphics[width=0.6\textwidth ]{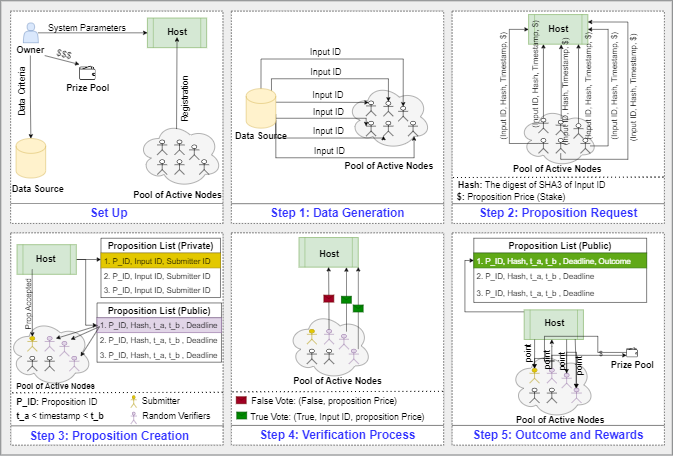}
\centering
\caption{An Overview of \tool}
\label{fig1}
\end{figure*}

\section{Preliminaries}
\label{background}

\subsection{Decentralized Oracle Networks (DONs)}
Smart contracts' limited ability to interact with the off-chain data is a significant limitation\cite{chainlinkblog}. An intermediary system called an oracle is necessary to establish a connection between smart contracts and off-chain data. DONs are used to input off-chain data into a blockchain to combat these challenges. A DON offers a solution to single points of failure in smart contracts by employing multiple data sources. Doing so ensures end-to-end reliability and makes it possible for high-value smart contracts to operate in low-trust environments.

\subsection{Related Work}
\label{related}
The challenge of smart contracts accessing the off-chain data has spawned numerous DON proposals. ChainLink\cite{chainl} aims to provide a cross-chain platform for internet-available data such as tamper-proof price data, automation functions, and external APIs. Neo Oracle Service\cite{neo} facilitates smart contracts to access external data by having oracle nodes designated by the committee fetch the data. Augur\cite{augur} and Gnosis\cite{gn} are prediction market platforms built upon Ethereum. TON\cite{TON} is a decentralized messaging application where messages are fed to the blockchain and verified by smaller agents. A Dag-Based Decentralized Oracle Model\cite{DAG} is a decentralized oracle system that validates the input data, such as the current Bitcoin price, from a centralized source by using a decentralized directed acyclic graph mode. Distributed blockchain Price Oracle\cite{dbpo} is a distributed system that provides exchange rates of digital assets to smart contracts. This network is used as a safety for decentralized finance applications. Astraea\cite{astrea} is a DON based on a voting game that decides the truth or falsity of propositions. Finally, A Smart Contract System for Decentralized Borda Count Voting\cite{borda} is a self-tallying decentralized e-voting protocol for a ranked-choice voting system based on Borda count.

With all these designs, decentralized solutions have yet to be presented for a counting system that can scale high. As mentioned in the problem statement part of Section ~\ref{problemstatement}, one needs to make a significant effort to use these systems as a DON for counting systems.

\begin{table}[t]
\centering
\begin{tabular}{|@{\hskip 2pt}p{0.58in}|p{1.8in}|p{0.5in}@{\hskip 2pt}|}
\multicolumn{1}{@{}l}{\textbf{Term}} & \multicolumn{1}{@{}l}{\textbf{Definition}} & \multicolumn{1}{@{}l}{\textbf{Set by}}\\
\hline

owner & the party interested in obtaining the result of the counting system  & network \\
\hline

host & platform storing all information & owner/ network  \\
\hline

player & participating nodes & network \\

\cline{2-3}

\multicolumn{1}{|r|}{$\sim$ ID} & unique identifier of each player & host\\

\cline{2-3}

\multicolumn{1}{|r|}{$\sim$ point} & the points accumulated
for each vote or proposition submission & host \\

\hline

proposition & a boolean request to be validated & submitter \\

\cline{2-3}

\multicolumn{1}{|r|}{$\sim$ deadline} & the deadline to vote for a proposition & owner/ host \\

\cline{2-3}

\multicolumn{1}{|r|}{$\sim$ pool} & the list of available propositions to be verified & host \\

\cline{2-3}

\multicolumn{1}{|r|}{$\sim$ price} & the stake used when players vote for propositions & owner \\

\hline

submitter & the player who submits a proposition to be verified & players/ host \\

\hline

verifier & randomly chosen player to verify a proposition & host \\

\hline

input ID & unique identifier of the input data & data source  \\
\hline

prize pool & financial reward and incentive & owner/ players \\

\hline

\end{tabular}
\vspace{0.4cm}
\caption{Summary of the terminology.}
\vspace{-0.8cm}
\label{table0}
\end{table}

\section{Description of \tool}

\subsection{System Assumptions}
In our design, we made the following assumptions:
\begin{enumerate}
    
    \item pseudo-random number generation is possible.
    
    \item A proposition $p$ has a truth value of either True or False.
    
    \item There are $N$ nodes numbered from $1$ to $N$. For each $p$ and each node $i \in [1, N]$, let  $b_i(p) \in \{T, F\}$ denote the belief of $i$ regarding the truth value of $p$. Each $i$ has an accuracy $q_i(p) \in [0, 1]$ that is, informally, the probability that $i$ is correct about their vote to a given proposition.

    \item Each node's belief is independent of others and honest nodes always vote truthfully.
    
\end{enumerate}

\subsection{Setup and Terminology}
The terminology summary is provided in Table ~\ref{table0}. This system requires an \emph{owner} and a \emph{host} to set the system parameters. The host is responsible for storing the data regarding \emph{propositions} and the information of the nodes. The nodes (oracles) are the validators of the system, similar to the miners in a blockchain network. The host is responsible for choosing a random selection process to ensure proper distribution of randomness across all nodes. The owner is the party interested in obtaining the result of the counting system. The owner provides the system with the initial funds and is responsible for defining the criteria for the counting object and other system parameters. The host and owner could be constructed using a separate smart contract as the host and the contract owner or a different contract as the owner.

Each event or object being counted is assigned a unique identifier, known as the Input ID. The system is based on a set of \emph{propositions} with covert truth values and associated price amounts, represented by \emph{proposition price}, used in a voting game. Each proposition has a termination condition by which the proposition outcome is decided. The owner sets this termination condition. Generally, the termination condition is the \emph{proposition deadline}, a deadline by which all the votes must be submitted. The system's fund or the \emph{prize pool} is the accumulation of the initial fund provided by the owner and the stakes collected from the dishonest players. The propositions are stored in the propositions Pool.

Each player (node) in the system is given two entities of \emph{player ID} and \emph{player point}. The player ID is a unique identification of the player (i.e., its blockchain address), and the player point is the points accumulated by the player for submitting a Vote or proposition. The final reward of each player is decided based on the player points. Finally, each player can have two roles in this system: the \emph{submitter} that submits a proposition to be verified by the system, and the \emph{verifier} that verifies the validity of each assigned proposition.

\begin{figure}[t]
\centering
\includegraphics[width=0.8\linewidth]{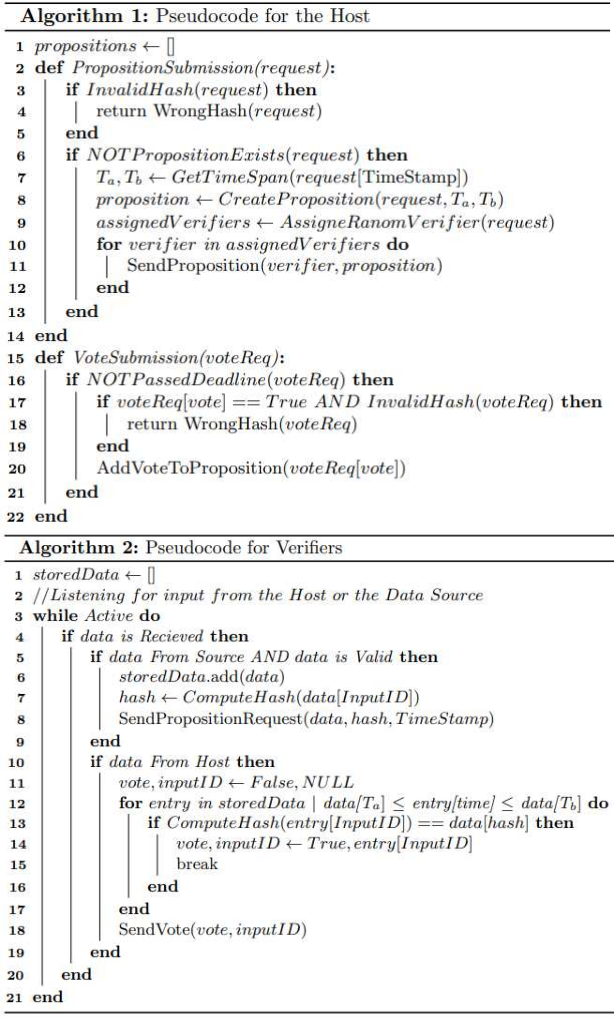}
\caption{Pseudocode for the Host and Verifiers}
\label{algs}
\end{figure}

\subsection{System Description}
\textbf{Step 1: } The data source sends the data regarding the object being counted to all the active players. Each player stores this data independently. The players know the criteria for counting objects ahead of time. For example, if the system counts the number of visitors to a particular website, the visitor's IP address is the Input ID, and the players can receive the required information using invisible pixels. 

\textbf{Step 2: }Once the players observe a valid event, they submit a proposition request to the host and pay the proposition price as a stake. This can be done by calling the host smart contract. The proposition request should contain the Input ID, Timestamp, player ID and Hash(Input ID) which is the digest of the cryptographic hash function of the Input ID.

\textbf{Step 3: }Once the host receives a valid proposition request, it generates a proposition with a new format, adds it to the \emph{proposition pool}, and assigns the proposition to some random players. These selected players become the verifiers of the proposition. The generated proposition by the host does not contain the Input ID; instead, it contains the Hash(Input ID). Also, the Timestamp is converted to  $T_a$ and $T_b$, where the Timestamp is the midpoint of $T_a$ and $T_b$ to account for any network delays. If needed, a second identifier (e.g., Ad ID) can be added to boost uniqueness within this time frame. The proposition also contains the proposition deadline, the time the assigned verifiers must submit their votes. The player who submitted the proposition request becomes the submitter of the proposition. The propositions are publicly available to all players, but only the assigned verifiers can vote. The owner decides the number of assigned verifiers for each proposition and is fixed for all the propositions.

\textbf{Step 4: } The assigned verifiers calculate the hash of all the Input IDs in the period of $T_a$ and $T_b$. If they find the corresponding hash as the Hash(Input ID) in the proposition, they vote True and submit the Input ID; otherwise, they vote False. If the assigned verifiers do not submit their votes by the deadline, it is regarded as voting False. The assigned verifiers must also pay the proposition price as a stake. Using hash functions makes it difficult for the verifiers to guess the correct Input ID given the Hash(Input ID) if they do not have the correct information in their records. On the other hand, it is relatively cheap to find the corresponding input ID from the limited data received. Therefore, when they vote True, they provide the input\_ID as proof of honesty.

\textbf{Step 5: } The majority of the votes decide the outcome of the proposition. The verifiers in the majority get rewarded, and those in the minority get penalized. If the owner requires an even number of verifiers for each proposition, the submitter's vote is counted as a True vote to break the tie.

Choosing a fixed random number of verifiers gives all propositions an equal number of verifiers; hence, the results are similar in accuracy. It also reduces the risk of Sybil attacks by ensuring players do not team up to pick one proposition to vote. The overview of \tool\ is shown in Figure~\ref{fig1}, and Figure ~\ref{algs} shows the algorithms for stpng 2 to 5.

\subsection{Reward and Penalty}

Whenever a player submits a proposition request or votes True, the host verifies that the provided Hash(Input ID) and ID correspond to another. If incorrect, the player loses the stake and receives two negative points for intentionally submitting a wrong format. This makes security attacks such as Denial of Service (DoS) expensive.

If most verifiers vote True for a proposition, the submitter receives four positive points, the verifiers that voted True receive one positive point, and the verifiers that voted False lose their stake and receive one negative point. If most verifiers vote False for a proposition, the submitter receives two negative points and loses the stake. The verifiers that voted True or False do not get any points. If a verifier votes True with the correct ID but remains in the minority, the wasted computational effort/cost is the penalty. Therefore, submitters take a higher risk and receive a higher reward for submitting a correct proposition. This is because one of the main factors determining the quality of the network is the quality of the submitted proposition requests.

Players receive a financial reward after a termination condition is met. This condition can be a fixed period, a number of propositions, or a certain number of Points. If the player point of a node is over zero, they get a financial reward based on their accumulated points using the below formula. In this formula, $p_{v_i}$ is the player point of a given player, ($u(V_j)$ is a Heaviside function, and $\sum_{j=1} p_{v_j} u(V_j)$ is the summation of the player points of all the players with positive points. 

\[\text{Verifier Prize} = \text{Prize Pool} \cdot \frac{ p_{v_i}}{\sum_{j=1} p_{v_j}u(V_j)} \] 

Each player's point gets reset after they are paid. Providing a financial reward based on a point system has several benefits. First, it motivates the verifiers to partake in the game continuously. It also incentivizes continuous honest behavior. If a verifier loses some points, they put more effort into the next propositions, and if a verifier gains some points, they are motivated to continue the honest behavior. It reduces the number of monetary transactions in the system, helps prevent lazy voting, and reduces Sybil attacks.
Mathematically, if we have a total of $N$ nodes and we choose $n$ random verifiers for every proposition, a Sybil attack is successful if the attackers gain $n/2$ verifiers for every proposition. Therefore, since verifiers are picked at random in \tool's design, the attack is only successful if the adversary gains $N - n/2$ nodes. The success rate becomes significantly low if $N$ is relatively larger than $n$. Furthermore, the attack is partially successful if the attackers gain $n/2$ verifiers for at least 50\% of the propositions. In this case, the attackers must gain at least 51\% of the total nodes. This means that while gaining 51\% of the power generally means a successful Sybil attack in most other blockchain applications, it only results in partial success in the case of \tool\ due to its design.

\subsection{Limitations of \tool}
\tool \ is a DON designed for counting systems or any system requiring precise counting, and it may not be suitable for most other applications. Moreover, like most other blockchain or decentralized applications, the quality of the network relies on the diversity and the number of available nodes. Additionally, CountChain requires an owner to set certain parameters. This limitation can be removed if, instead of an owner, the majority of the participants, a bigger blockchain like Ethereum or a smart contract sets these parameters.

\section{Game Theory: The Verifier's Dilemma}
\label{game}

Lazy voting and Nash Equilibrium (NE) are essential concepts when designing a DON. Lazy voting is when players blindly vote True or False, knowing either option is more likely to win. This increases dishonest voting and can significantly make the outcome inaccurate. Therefore, a DON needs to be designed to prevent lazy voting.

In \tool, lazy voting True is impossible as verifiers must provide proof when voting True. On the other hand, the best reward for voting False is not to get penalized. This means lazy voting False has a high risk (losing money and points) and no reward. Therefore, verifiers are only incentivized to vote False if they perform the computation and ensure that the correct vote is False. Also, not voting is considered voting False. This ensures inactive players are not incentivized to stay in the game. Moreover, submitting a False proposition request gets financial and point punishment; hence, submitters need to make an effort to ensure their submission is valid.

A Nash equilibrium (NE) is a situation in game theory in which no player can benefit by changing their strategy, given that all other players' strategies remain the same. Analyzing NE ensures that the DON design incentivizes honest behavior.

In \tool's design, in each round, every verifier needs to make two decisions. These are depicted in \Cref{fig:extensive_form_game}.

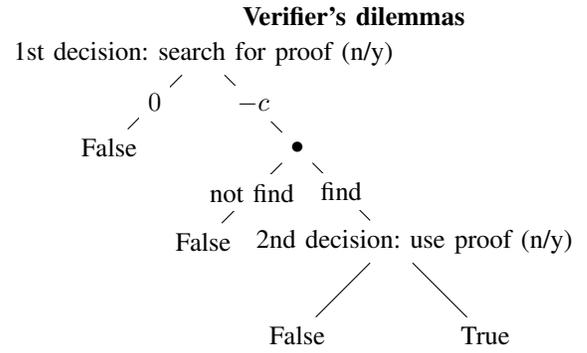
\begin{figure}[t]
    \centering
       \begin{tikzpicture}
        \def\dx{1.25}
        \def\dy{1.25}
            \node (t0) at (2,0.5) {\textbf{Verifier's dilemmas}};
            \node (n00) at (0,0) {1st decision: search for proof (n/y)};
            
            \node (n10) at (-1*\dx,-1*\dy) {False};
            \node (n11) at (1*\dx,-1*\dy) {$\bullet$};
            
            \node (n20) at (0,-2*\dy) {False};
            \node (n21) at (2*\dx,-2*\dy) {\ \ \ \ \ 2nd decision: use proof (n/y)};

            \node (n30) at (1*\dx,-3*\dy) {False};
            \node (n31) at (3*\dx,-3*\dy) {True};

            \draw (n00) to node[pos=0.5,fill=white] {$0$} (n10);
            \draw (n00) to node[pos=0.5,fill=white] {$-c$} (n11);
            
            \draw (n11) to node[pos=0.5,fill=white]{not find} (n20);
            \draw (n11) to node[pos=0.5,fill=white]{find} (n21);

            \draw (n21) to node{} (n30);
            \draw (n21) to node{} (n31);

    \end{tikzpicture}
    \caption{ A verifier needs to make two decisions: whether to put effort with sunk cost $-c$ and search for proof and if they find proof, whether to use it and vote "True" or ignore it and vote "False". }
    \label{fig:extensive_form_game}
\end{figure}

\begin{table}[t]
\centering
    \setlength{\extrarowheight}{3pt}
    \begin{tabular}{cr|r@{\hskip 2pt}l|r@{\hskip 2pt}l|}
      & \multicolumn{1}{c}{} & \multicolumn{4}{c}{Majority}\\
      & \multicolumn{1}{c}{} & \multicolumn{2}{c}{True}  & \multicolumn{2}{c}{False} \\\cline{3-6}
      \multirow{2}*{Verifier}  & True & $[1,0]$, & $[1,0]$  & $[0,0]$, & $[0,0]$  \\\cline{3-6}
      & False & $[-1,-1]$, & $[1,0]$ & $[0,0]$, & $[0,0]$  \\\cline{3-6}
    \end{tabular}\vspace{0.4cm}
    \caption{Summary of payoffs. The entries in each bracket correspond to [Point, Stake].}
    \vspace{-0.7cm}
    \label{tab:bne}
  \end{table}
  
The first decision is whether or not to search for proof. Searching for proof costs $c>0$ for the verifier, i.e., the verifier earns $-c$ from searching, whereas doing nothing incurs a cost of $0$ where $c$ is the computational cost. If a verifier does not search for proof, they may only vote "False" according to the rules of the protocol. The same holds if they search for proof but do not find one. However, if they find proof, they face the second dilemma: whether to use this proof and vote "True" or ignore the proof and vote "False." \par

The utility of a verifier depends on its decision and the behavior of other verifiers. From \Cref{tab:bne}, it is immediate that first, it is optimal for each verifier to be in the majority and second, it is a weakly dominant strategy for a verifier to vote True whenever they possess a valid proof. If the majority votes True, the verifier is strictly better to vote True and if the majority votes False, the verifier is indifferent.\par
The game between verifiers may possess many Nash equilibria. For instance, if all (other) verifiers engage in lazy voting, then it is best for a remaining verifier to do nothing and vote False by default. The reason is that if every other verifier is voting False, then no proposition will be decided as True, and incurring the cost of searching for proofs will only result in a negative utility for an honest verifier. \par
The main question we need to answer is whether honest behavior, i.e., searching for proofs and submitting them (voting True) if one finds them, is a Nash equilibrium in the game. Assuming a fraction $p_T\in(0,1]$ of True propositions, this is answered affirmatively in \Cref{thm:nash} under the mild assumption that $2p_T>c$, where $c$ denotes the cost of searching for a valid proof (typically assumed to be very low since it only involves the calculation of a small number of hashes).

\begin{theorem}[Incentive Compatibility of \tool]\label{thm:nash}
Assume that a fraction $p_T\in(0,1]$ of propositions are True and that searching for a valid proof incurs a cost $c>0$ to the verifier. Then, if all other verifiers are behaving honestly, i.e., search for valid proofs and vote True whenever they find one, then it is best for a remaining verifier to also behave honestly if and only if $2p_T>c$.
\end{theorem}
\begin{proof}
Assuming that all verifiers behave honestly ensures that a fraction of $p_T$ propositions will be decided as True, i.e., all True propositions will be decided as True and all False propositions will be decided as False. Thus, according to \Cref{tab:bne}, the expected utilities of a single verifier from their possible actions $\{no-search, search-false, search-true\}$ are the following
\begin{align*}
    \mathbb{E}[\text{no-search}]&= p_T(-1)+(1-p_T)\cdot0 = -p_T\\
   \mathbb{E}[\text{search-false}]&= -c + p_T(-1)+(1-p_T)\cdot0= -p_T-c\\
    \mathbb{E}[\text{search-true}]&= -c + p_T\cdot1+(1-p_T)\cdot0= p_T-c
\end{align*}
Clearly, the action search-false is dominated by no-search, since $-p_T-c<-p_T$ for $c>0$. Thus, to determine a verifier's optimal action, we only need to consider the actions "search" and "vote True" whenever one possesses a valid proof. The expected utility of a verifier who selects "search-true" $x\in[0,1]$ fraction of time and "no-search" $1-x$ fraction of time is 
\begin{align*}
    \mathbb{E}[x]&= x\left[-c+p_T\cdot 1+(1-p_T)\cdot 0\right]\\&\phantom{=\,\,}+(1-x)\left[p_T\cdot(-1)+(1-p_T)\cdot0\right] \\&=x(2p_T-c)-p_T.
\end{align*}
This is increasing in $x$ whenever $2p_T>c$ as claimed. In particular, the verifier's utility is optimized for $x=1$ which corresponds to always honest behavior.  
\end{proof}

In general, assuming that $p$ of propositions are decided as True for any any $p\in[0,1]$ (not necessarily equal to $p_T$), then the statement of the \Cref{thm:nash} continues to hold, i.e., it is the best response to be honest as long as $2p>c$.

\section{Experimental Evaluation}
\label{exp}
\subsection{Preliminary and Experimental Setup}
As stated in section ~\ref{related}, the available designs require a major change to be used for counting systems. \tool \ is created to fill a pervasive need for a practical DON for counting systems. As a result, we performed multiple experiments to test the feasibility of the design of  \tool\ and determine the optimal system parameters. The experiments were run on an Azure E96ds v5 virtual machine with 96 VCPUs and 672 GiB memory on Ubuntu 20.04 OS.

\begin{figure}[t]
\centerline{\includegraphics[width= 0.7 \textwidth/2]{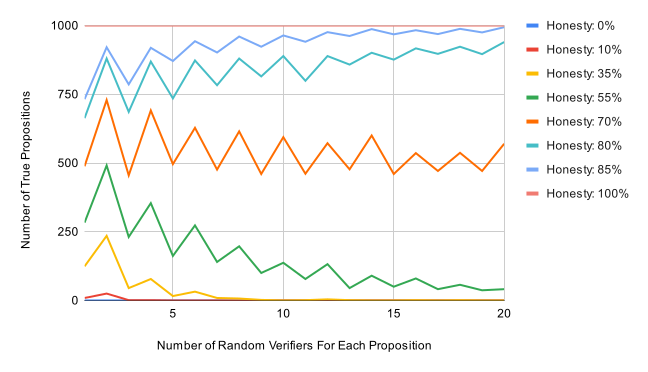}}
\centering
\caption{Optimal Honesty Rate and the Number of Verifiers}
\label{ex1}
\end{figure}

\subsection{Experiment 1: Optimal System Parameters}

This experiment aims to demonstrate the impact of system parameters. These parameters are each node's honesty rate, the number of chosen random verifiers for each proposition, and the total number of nodes. Honesty rate mimics the rate of accuracy or honesty of a node, as nodes may intentionally or unintentionally vote incorrectly for a proposition. 

In this experiment, We first varied the honesty rate for each node from 0\% to 100\%, in 5\% increments. 0\% means "not honest" and 100\% means "most honest". We changed the number of random verifiers assigned to each proposition from 1 to 20 with increments of 1 for each honesty rate. Each time, we generated and sent 1,000 random valid input data for a total of 420,000 valid data at the rate of 10 data per second. Ideally, we expected to see 1,000 True propositions for each scenario, so we used these values as the ground truth. The verifiers were picked at random from a pool of 100 available nodes. We used the following values for other parameters: proposition price :1, proposition deadline: 2 sec, proposition Delay: 1 sec, and Threshold Point: -5.

The summary of the results is shown in Figure~\ref{ex1}. As expected, the best outcome was when the honesty rate of all nodes was 100\%, whereas the worst result occurred when the rate was 0\%. Also, when the rate was 10\% or higher, all the propositions were successfully raised, and in the case of 5\%, only a small number of propositions were not submitted. This is because all nodes have access to the data and can initiate a proposition request. Therefore, this result has not been included in the charts above to avoid duplicity.  

The first noticeable accuracy loss becomes evident between 85\% and 80\%. This is because as the honesty rate of all nodes decreases, the chance of the majority being on the inaccurate side increases. As shown, a precision level of up to 99\% can be achieved with the honesty rate of 85\%, whereas the maximum accuracy attainable is 95\% for the rate of 80\%. Hence, it is recommended to configure the system to ensure all nodes possess an honesty rate of at least 85\%. This is achieved by filtering out nodes that have been proven unreliable. Additionally, for this honesty rate, the highest level of precision was observed for 20, 18, and 14 verifiers, resulting in 994, 988, and 987 True propositions, respectively. Thus, the recommended optimal number of chosen verifiers per proposition is fourteen, which reasonably balances accuracy and resource allocation.

Furthermore, for the honesty rate of over 80\%, increasing the number of verifiers improves the precision of the outcomes. This is the opposite when the honesty rate is below 60\%. This is because most verifiers are honest if the honesty rate is high for a given proposition. Additionally, an even number of verifiers yields higher accuracy. This is because the submitter's vote is only considered and counted as True in the event of a tie in the overall vote, which can only occur when the number of verifiers is even. Nevertheless, this observation highlights that this rule improves the precision of the system's outcomes.

\begin{figure}[t]
\centerline{\includegraphics[width= 0.65 \textwidth/2]{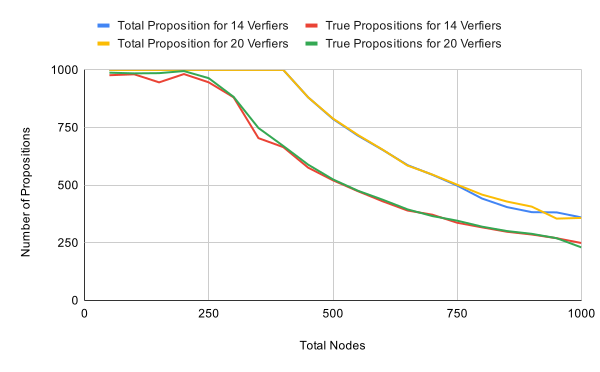}}
\centering
\caption{Optimal Number of Nodes}
\label{ex2}
\end{figure}

Then, we varied the number of total nodes from 50 to 1,000 in increments of 50. Fourteen verifiers were randomly picked for each proposition from the pool of nodes. The honesty rate for all the nodes was set to 85\%. We also repeated the experiment by picking 20 verifiers to compare the results. 

As shown in Figure~\ref{ex2}, increasing the number of nodes from 50 to up to 250 can improve the accuracy in both cases. However, the performance starts to decrease when there are more than 250 nodes. When the total number of nodes is more than 400, the number of raised propositions decreases as the number of nodes increases. Also, when the total number of nodes is more than 250, the number of True propositions decreases as the number of nodes increases. 

In our setup, 200 nodes yield the most favorable results. Furthermore, when fourteen verifiers were chosen from a pool of 200 nodes, 982 propositions out of the raised 1,000 were correctly decided as True. Meanwhile, this number is 995 when the number of verifiers is 20. This indicates a less than 1\% increase in accuracy for using six additional verifiers per proposition. These results demonstrate that fourteen verifiers from two hundred nodes achieve the optimal balance between accuracy and resource allocation.

\textbf{Performance Analysis: } Our analysis shows that every extra 50 nodes results in an extra 0.6 to 1 \% CPU and RAM usage and no noticeable disk usage. The main bottleneck is the allocated heap memory for the program. Even when set to max, the program can reach the maximum allocation. Thus, this is mainly a limitation of the available hardware resources than the design. Hence, this limitation should be handled if a distributed infrastructure is used as needed for a blockchain design. 

Also, On average, the latency (neglecting network delays) is about 116 ms, indicating that \tool\ can scale effectively. This is because the \tool's execution involves fixed data and proposition formats with adaptable system parameters and minimal storage costs due to small data sizes and short data lifespans with the option to transfer data to low-cost blockchains subsequently. The consensus protocol is efficient in time, cost, and resource usage, primarily requiring minimal computations from selected verifiers.

\begin{figure}[t]
\centerline{\includegraphics[width= 0.65 \textwidth/2]{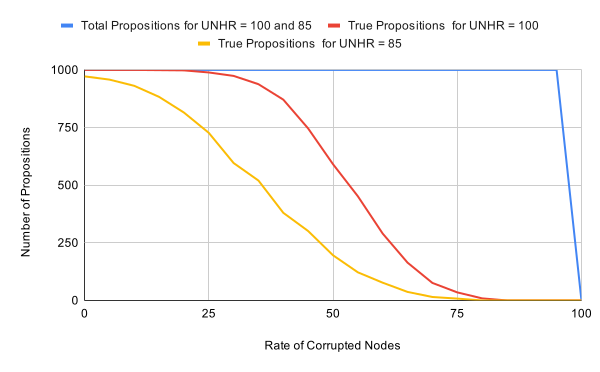}}
\centering
\caption{Sybil Attack Success Rate. UNHR = uncorrupted node honesty rate}
\label{ex3}
\end{figure}

\subsection{Experiment 2: Sybil Attack Success Rate}
\label{sybil}
A Sybil Attack is when an entity has numerous fake identities for malicious motives\cite{syb}. This experiment aims to simulate a Sybil attack and find the success rate of this attack on \tool. In this experiment, fourteen verifiers were picked from the pool of two hundred nodes. A Sybil attack is considered fully successful if the disruptive nodes either prevent proposition creation or cause all generated propositions to be incorrectly marked as False. 

In this experiment, we set the honesty rate of the corrupted nodes to 0\% and varied the rate of corrupted nodes from 0 to 100\% (in increments of 5). Generally, when determining the success rate of a Sybil attack, the honesty of all the honest nodes is considered 100\%. Therefore, we decided to have two sets of test cases with the honesty rate of honest nodes as 85\% and 100\%. For example, if there are 200 nodes where half are corrupted, 100 random nodes would have an honesty rate of 0, and the honesty rate for the rest would be 85\% and 100\%. The rest of the parameters are similar to the previous experiment.

As shown in Figure~\ref{ex3}, when the honesty rate of uncorrupted nodes was 100\%, even when the corrupted nodes gained 50\% of the network, over half of the propositions were correctly decided as True. The attack was not entirely successful until corrupted nodes gained at least 85\% of the network. Meanwhile, changing the honesty rate of uncorrupted nodes to 85\% significantly increased the attack's success rate. In this case, the attackers needed to gain 70\% of the network to launch the attack successfully. Another notable observation is that all the propositions were successfully raised in both cases until the corrupted nodes completely controlled the network. This is because all the data is sent to all the nodes, and as long as one honest node is active, they can raise the proposition. Although this does not prevent the attack, it can be used as reference data for the future. 

The results show that \tool\ performs well against Sybil attacks due to many factors included in the design. First, only a selected number of nodes can vote on a proposition. In most other DONs, all nodes can participate freely, so if an entity achieves the majority of the power in a network, they are guaranteed to succeed in the attack. In \tool, the attackers need to gain more than the majority for guaranteed success, as shown in this experiment. Also, the point and stake system increases the cost of DoS and Sybil attacks. Finally, nodes accumulating a few negative points are banned from the system, making a Sybil attack harder.

\textbf{Sybil Attack Prevention:} PoW, PoS, and DPoS are the main sybil control mechanisms used in Blockchain\cite{OSCMS}. PoW safeguards against Sybil attacks by requiring nodes to employ a limited resource to create new blocks.\cite{bitcoin}. However, this design has a high energy consumption\cite{book}. 

PoS assigs voting power based on the number of tokens locked to a node. Therefore, attacking the network becomes costly\cite{rad}. However, the more stakes a node has, the more blocks it can create, making it more powerful. Since the number of coins/tokens is limited, once a node gains 51\% of the total stake, it achieves the majority of the power that cannot be revoked later\cite{book}. 

DPoS is a type of PoS where token holders can delegate their tokens to node-runners, which increases the tokens required for a Sybil attacker. The main limitation of this system is that it can become centralized more easily when there are fewer delegates, as a malicious coalition of at least 51\% of delegates could control the network more easily.\cite{cryptoc}. 

\tool \ minimizes sybil attacks using system parameters that are adjustable by the host based on the needs. Also, few verifiers are randomly chosen from a pool of available verifiers, which, along with the point and stake system, makes sybil attacks more expensive. Finally, an authentication step can be added for each node. 

For the suggested parameters, if 66 out of 200 nodes (33\%) are dishonest, the probability of at least 8 dishonest out of 14 ($\ge 51\%$) random verifiers is 5\%:

\centerline{
$ \frac{\binom{66}{8} \binom{134}{6} + \binom{66}{9} \binom{134}{5} + \cdots + \binom{66}{14} \binom{134}{0}}{\binom{200}{14}} $
}

Plugging in other numbers confirms that the finding is aligned with Figure ~\ref{ex3} (red line).

\section{Conclusions and Future Work}
\label{Conclusion}

One of the main barriers to the mass adoption of blockchain is its limitation to interact with off-chain data. This paper outlines the design and implementation of CountChain, a blockchain oracle tailored for counting systems. This opens up many opportunities for billion-dollar industries, such as online advertising, which rely on counting systems, to adopt blockchain technology. In the future, we plan to expand our design to an industry-standard platform for commercial use as it possesses a tangible commercial potential that can improve decentralized ecosystems.

\bibliographystyle{IEEEtran}
\bibliography{Cover}

\end{document}